\documentclass{amsart}
\usepackage[dvips]{graphicx}
\usepackage{graphicx}
\usepackage{amscd}
\usepackage{amsxtra}
\usepackage{enumerate}
\usepackage{amsfonts}
\usepackage{amsmath}
\usepackage{amssymb}
\usepackage{amsthm}

\newtheorem{theorem}{Theorem}[section]
\newtheorem{corollary}[theorem]{Corollary}

\theoremstyle{definition}
\newtheorem{definition}[theorem]{Definition}

\newtheorem{example}[theorem]{Example}

\theoremstyle{remark}

\numberwithin{equation}{section}

\usepackage{graphicx}

\newcommand{\Iidentity}{\hbox{\upshape \small1\kern-3.3pt\normalsize1}}
\def\blfootnote{\xdef\@thefnmark{}\@footnotetext}

\begin{document}

\title[An extension of Wiener integration with the use of operator theory]{An extension of Wiener integration with the use of operator theory}

\author{Palle E. T. Jorgensen}
\address{Department of Mathematics, The University of Iowa, Iowa City, IA52242, USA}
\curraddr{}
\email{jorgen@math.uiowa.edu}
\urladdr{http://www.math.uiowa.edu/~jorgen}
\thanks{Work supported in part by the U.S. National Science Foundation}

\author{Myung-Sin Song}
\address{Department of Mathematics and Statistics, Southern Illinois University Edwardsville, Edwardsville, IL62026, USA}
\curraddr{}
\email{msong@siue.edu}
\urladdr{http://www.siue.edu/~msong}

\date{December, 2008}





\begin{abstract}
With the use of tensor product of Hilbert space, and a diagonalization 
procedure from operator theory, we derive an approximation formula for a 
general class of stochastic integrals. Further we establish a generalized 
Fourier expansion for these stochastic integrals. In our extension, we 
circumvent some of the limitations of the more widely used stochastic 
integral due to Wiener and Ito, i.e., stochastic integration with respect 
to Brownian motion. Finally we discuss the connection between the two 
approaches, as well as a priori estimates and applications.
\end{abstract}

\maketitle \tableofcontents

\section{Introduction}
\label{sec:1}

Recently there has been increase in the number of applications of 
stochastic integration and stochastic differential equations (SDEs). In 
addition to the traditional applications in physics and dynamics, stochastic 
processes have found uses  in such areas as option pricing in finance, 
filtering in signal processing, computations biological models. This fact 
suggests a need for a widening of the more traditional approach centered 
about Brownian motion $B(t)$ and Wiener's integral.

Since SDEs are solved with the use of stochastic integrals, we will focus 
here on integration with respect to a wider class of stochastic processes 
than has previously been considered. In evaluation of a stochastic integral 
we deal with the term $dB(t)$ by making use of the basic properties of 
Brownian motion, such as the fact that $B(t)$ has independent increments. If 
instead $X(t)$ is an arbitrary stochastic process, it is not at all clear 
how to make precise a stochastic integration with respect to $dX(t)$. We 
will develop a method, based on a Karhunen-Lo\`{e}ve diagonalization, for 
doing precisely that. 

The theory of stochastic integrals is well developed, see e.g., 
\cite{Ku06, ItMc65}. For many applications, such as the solutions to 
stochastic differential equations in physics and finance, it is important 
to have tools for evaluating integrals with respect to $dB$ where $B$ is 
Brownian motion.  The reason for the technical issues involved in the 
computation of stochastic integrals can be understood this way:  A naive 
approach runs into difficulties, for example because the length of Brownian 
paths is infinite, and because Brownian paths are discontinous (with 
probability one). Wiener and Ito offered a, by now, well known way around 
this difficulty. The idea of Wiener in fact is operator theoretic: It is to 
establish the value of an integral as a limit that takes place in a Hilbert 
space of random variables. This is successful because of the existence of 
an isometry between this Hilbert space on the one hand and a standard 
$L^{2}-$Lebesgue space on the other.

In this paper we extend this operator theoretic approach to a much wider 
class of stochastic integrals, i.e., integration with respect to $dX$ where 
$X$ belongs to a rather general class of stochastic processes. And we give 
some applications.

In the proof of our theorem we make use of a result from two earlier papers 
\cite{JS07, JS08} by the coauthors. The idea is again operator theoretic, 
and it is based on an application of von Neumann's spectral theorem to an 
integral operator directly associated with the process $X$ under 
consideration.

While the applications of stochastic integrals to physics (e.g., \cite{BC97}, 
and their interplay with operator theory (e.g., \cite{AL08a})are manifold, the 
idea of exploring and extending the scope in the present direction appears to 
be new. The need for such an extension is convincing:  For example, physical 
disturbances or perturbations will typically take you outside the particular 
path-space framework where the theory was initially developed. 

Earlier approaches to stochastic integrals with reproducing kernels include 
\cite{AL08a, AAL08, AL08b} and operator theory \cite{JM08}; and papers 
exploring physical ramifications: \cite{BC97, BDGJL07, Hu07a, Hu07b}. Although 
the papers cited here with physics applications represent only the tip of an 
iceberg!

\section{Notation and Definitions}
\label{sec:2}

To make precise the operator theoretic tools going into our construction, we 
must first introduce the ambient Hilbert spaces. Since stochastic integrals 
take values in a space of random variables, we must specify a fixed 
probability space $\Omega$, with sigma algebra and probability measure. In 
the case of Brownian motion, the probability space amounts to  the standard 
construction of Wiener and Kolmogorov: The essential axiom in that case is 
that all finite samples are jointly Gaussian, but we will consider general 
stochastic processes, and so we will not make these restricting assumptions 
on the sample distributions and on the underlying probability space. For more 
details on this case, see section \ref{sec:4}.

The kind of integrals we consider presently are stated precisely in 
Definition \ref{D:2.1}, eq (\ref{eq2.5}) below. I particular, initially we 
consider only functions of time in the integrant, so  $f(t) dX_{t}$. 
When the stochastic process  $X$ is given, we show (Theorem \ref{T:3.1})  
that the corresponding integrals live in  a Hilbert space which is a direct 
sum of standard Lebesgue Hilbert spaces carrying the function $f(t)$. In 
the case of Brownian motion, we show (Example \ref{E:4.1}) that the direct 
sum representation then only has one term. 

We now list the symbols and the terminology.
\begin{itemize}
\item $L^{2}$: an $L^{2}$-space. 
\item $L^{2}(\mathbb{R})$: all $L^{2}$-functions on $\mathbb{R}$. 
\item $J \subset \mathbb{R}$ a finite closed interval. 
\item $L^{2}(J)$: $L^{2}$ with respect to the Lebesgue measure restricted
      to $J$. 
\item $(\Omega, \mathcal{S}, P)$: a fixed probability space. 
\item $\Omega$: sample space. 
\item $\mathcal{S}$: some sigma algebra of subsets of $\Omega$. 
\item $P$: a probability measure defined on $\mathcal{S}$. 
\item $L^{2}(J \times \Omega, m \times P)$: the $L^{2}$-space on 
      $J \times \Omega$ with respect to product measure $m \times P$ where 
      $m$ denotes Lebesgue measure.
\end{itemize}

\begin{enumerate}[(1) ]
\item $X_{t}: \Omega \to \mathbb{R}$, $t \in \mathbb{R}$ a stochastic 
      process, 
\item $X_{t}(\omega):=X(t,\omega)$, $t \in \mathbb{R}$, 
      $\omega \in \mathbb{R}$. 
\item $E(\cdot):=\int_{\Omega} \cdot dP$: the expectation with respect to 
      $(\Omega, \mathcal{S}, P)$. \\
      Restricting Assumptions:
      \begin{enumerate}[(i) ]
      \item $X \in L^{2}(J \times \Omega, m \times P)$ for all finite 
            intervals $J \subset \mathbb{R}$. 
      \item $(s, t) \longmapsto E(X_{s} X_{t})$ is continuous on 
            $J \times J$. ($E(X_{t})=0$). 
      \item For all $J$, and all $s \in J$, the function, 
        \begin{equation}
        \label{eq2.1}
          t \longmapsto E(X_{s} X_{t}) 
        \end{equation} 
            is of bounded variation.
      \end{enumerate}
      For $J \subset \mathbb{R}$ fixed, we consider partitions 
      \begin{equation}
      \label{eq2.2}
        \pi: t_{0} < t_{1} < \cdots < t_{n-1} < t_{n} =: t, \quad 
        J=[t_{0}, t];
      \end{equation}
      and we set 
      \begin{equation}
      \label{eq2.3}
        |\pi|:= \max_{i}(t_{i+1}-t_{i}), \quad \mbox{ and} \quad 
        \Delta t_{i}:=t_{i+1}-t_{i},  \quad 0 \leq i < n.
      \end{equation}

      If $f:\mathbb{R} \to \mathbb{C}$ is continuous, we set
      \begin{equation}
      \label{eq2.4}
        S_{\pi}(f, X):= \sum_{i=0}^{n-1}f(t_{i})
        (X_{t_{i+1}}-X_{t_{i}}).
      \end{equation}
      \begin{definition}
      \label{D:2.1}
        By a stochastic integral, we mean a limit
        \begin{equation}
        \label{eq2.5}
          \lim_{|\pi| \to 0} S_{\pi}(f, X)=: \int_{t_{0}}^{t}f(s)dX_{s}
        \end{equation}
      \end{definition}
\end{enumerate}

We now turn to questions of existence of this limit for a rather general
family of stochastic processes $X_{t}$; see (i)-(iii) below.

\section{Statement of the Main Theorem}
\label{sec:3}

When the stochastic process $X$ is given, we proved in Theorem \ref{T:3.1} 
that the corresponding integrals live in a Hilbert space which is a direct 
sum of standard Lebesgue Hilbert spaces carrying the function $f(t)$. In 
the case of Brownian motion, we now show that the direct sum representation 
then only has one term. Yet the method from section \ref{sec:3} still offers 
a Fourier decomposition of the Wiener integration.

\begin{theorem}
\label{T:3.1}
  Let $(\Omega, \mathcal{S}, P)$ be given as above, and let $X$ be a 
  stochastic process satisfying conditions $(i)-(iii)$.  Let $f$ be given 
  and continous.
  \begin{enumerate}[(a) ]
    \item Then the stochastic integral $\int_{t_{0}}^{t}f(s)dX_{s}$ exists 
          and is in $L^{2}(\Omega, P)$.
    \item There is a family of bounded variation functions 
          $\varphi_{1}, \varphi_{2}, \cdots $ and numbers 
          $\lambda_{1}, \lambda_{2}, \cdots $ satisfying the following 
          conditions:
          \[
            \lambda_{1} \geq \lambda_{2} \geq \cdots > 0, \quad 
            \lambda_{k} \to 0
          \]
          in fact $\sum_{k}\lambda_{k} < \infty$,
         such that
         \begin{equation}
         \label{eq3.1}
           E(\left\vert \int_{t_{0}}^{t}f(s)dX_{s}\right\vert^{2})
           = \sum_{k=1}^{\infty}\lambda_{k} \left\vert 
           \int_{t_{0}}^{t}f(s)d\varphi_{k}(s)\right\vert^{2}
         \end{equation}
         where the terms $\int_{t_{0}}^{t}f(s)d\varphi_{k}(s)$ refer to 
         Stieltjes integration. 
    \item If an interval $J$ is chosen such that $t_{0}$ and $t \in J$, and 
          if $f$ has a weak derivative $f'$ in $L^{2}(J)$, then the following 
          estimate holds for the RHS in (\ref{eq3.1}):
          \begin{equation}
          \label{eq3.2}
            \mbox{ RHS } \leq \mbox{ ``Const" } + 
            \lambda_{1}\int_{t_0}^{t}|f'(s)|^{2}ds
          \end{equation}
          where ``Const" depends on certain boundary conditions, and where 
          $\lambda_{1}$ is the maximal eigenvaluel see (\ref{eq5.4}) below.
  \end{enumerate}
\end{theorem}

In the next corollary, we stay with the assumptions from the theorem; 
in particular $X_{t}$ is a stochastic process subject to conditions 
(i)-(iii), and a compact interval $J$ is fixed.

\begin{corollary}
\label{C:3.2}
Covariance relations:
\begin{itemize}
  \item $E(X_{s} X_{t})
   =\sum_{k=1}^{\infty}\lambda_{k}\overline{\varphi_{k}(s)}\varphi_{k}(t)$
  \item $E(X_{t}^{2}) =\sum_{k=1}^{\infty}\lambda_{k}|\varphi_{k}(t)|^{2}$
  \item Dependency of increments:  If $s<t<u$ in $J$, then
        \begin{align*}  
          &E((X_{t}-X_{s})(X_{u}-X_{t}))  \\
          &= \sum_{k=1}^{\infty}\lambda_{k}
          (\overline{\varphi_{k}(t)}\varphi_{k}(u)
          -\overline{\varphi_{k}(s)}\varphi_{k}(u) + 
          \overline{\varphi_{k}(s)}\varphi_{k}(t)-|\varphi_{k}(t)|^{2})
        \end{align*} 
  \item $E((X_{t+\Delta t} - X_{t})^{2})=\sum_{k=1}^{\infty}\lambda_{k}
    |\varphi_{k}(t+\Delta t)-\varphi_{k}(t)|^{2}$
\end{itemize}
\end{corollary}

\section{An Application}
\label{sec:4}

In this section we restrict the setting of Theorem \ref{T:3.1} to the 
special case when  $X = B$, i.e., to the special case of integration 
with respect to Brownian motion. We then work out the eigenfunctions 
and eigenvalues for the covariance operator. It turns out to be the 
familiar Fourier basis. Actually there is a choice of bases depending on 
boundary conditions. A choice of the Dirichlet conditions yields the ONB 
of the sine functions. We further show that when the eigenvalue expansion 
is summed (using orthogonality) we then arrive at the familiar Wiener-Ito 
formula.

\begin{example}
\label{E:4.1}
  $X=B=$ Brownian motion.
  \begin{itemize}
    \item $(\Omega, \mathcal{S}, P)$ Gaussian space;
    \item $\Omega$ a space of functions, 
          $\mathcal{S}:=\langle \text{cylinder sets}\rangle$ 
          $\sigma-$ algebra; the sigma algebra generated by the cylinder-sets.
    \item $J=[0,1]$;
    \item $E(X_{s} X_{t})=\min(s, t)=:s \wedge t$;
    \item $X_{t}(\omega)=\omega(t)$, $\omega \in \Omega$; 
    \item $E((X_{t+\Delta t}-X_{t})^{2})=\Delta t$.
  \end{itemize}
  We now show that the known formula
  \begin{equation}
  \label{eq4.1}
    E(\left\vert \int_{0}^{t}f(s)dB_{s}\right\vert^{2})= 
    \int_{0}^{t}|f(s)|^{2}ds
  \end{equation}
  follows from the theorem; and in particular from (\ref{eq3.1}).
 
  In the case of Brownian motion for the functions $\varphi_{k}$ we may take 
  \begin{equation}
  \label{eq4.2}
    \varphi_{k}(t)=\sqrt{2}\sin(k\pi t); \quad k=1,2, \cdots . 
  \end{equation}
  Note that
  \begin{equation}
  \label{eq4.3}
    \varphi_{k}(0)=\varphi_{k}(1)=0, \quad \forall k=1,2, \cdots ;
  \end{equation}
  and 
  \begin{equation}
  \label{eq4.4}
    \lambda_{k}=\frac{1}{(k\pi)^{2}}
  \end{equation}

  Set $t_{0}=0$ for simplicity.  Note that if 
  \begin{equation}
  \label{eq4.5}
    g(t):= \int_{0}^{1} t \wedge s f(s)ds,
  \end{equation}
  then 
  \[
    (\frac{d}{dt})^{2}g(t)=-f(t),
  \]
  so the eigenvalue problem
  \begin{equation}
  \label{eq4.6}
    \int_{0}^{1}t \wedge s f(s)ds=\lambda f(t)
  \end{equation}
  has the solution given by (\ref{eq4.2}) and (\ref{eq4.4}).  Note further that 
  (\ref{eq4.3}) is a choice of boundary conditions.

  To see that (\ref{eq4.1}) follows from (\ref{eq3.1}) in the theorem, we 
  proceed as follows; starting with the RHS in (\ref{eq3.1}) and using 
  $d(\sin(k\pi t)) = -k\pi \cos(k\pi t)dt$:
  \begin{align*}
    \sum_{k=1}^{\infty}\lambda_{k} \left\vert 
           \int_{0}^{t}f(s)d\varphi_{k}(s)\right\vert^{2} &=
    \sum_{k=1}^{\infty}\frac{2}{(k\pi)^{2}}\left\vert 
    \int_{0}^{t}f(s)d\sin(k\pi s)\right\vert^{2}  \\
    &\underset{by (\ref{eq4.3})}{=}  
      2\sum_{k=1}^{\infty}\left\vert \int_{0}^{t}f(s)\cos(k\pi s)ds
      \right\vert^{2}  \\
    &\underset{by Parseval's formula}{=} \int_{0}^{t}|f(s)|^{2}ds; 
  \end{align*}
  and the desired conclusion follows.
\end{example}

\section{Proof of Theorem \ref{T:3.1}}
\label{sec:5}

Here we give the details of proof of theorem \ref{T:3.1}. Since the proof is 
long, to help the reader our presentation is divided into two 
parts, A and B.

Part A is an outline of the steps in the proof itself, and part B contains 
the details arguments making up each part in the proof. Part A begins with 
the notation and the terminology, introducing an auxiliary selfadjoint 
operator, its matrix approximations, and its spectral resolution.

\subsection{Part A}
\label{partA}
\begin{itemize}
  \item Select a fixed interval $J:=[a,b]$, $a<b$.
  \item From the assumptions on the process $(X_{t})_{t \in \mathbb{R}}$ 
        note that the operator
        \begin{equation}
        \label{eq5.1}
          (T_{J}f)(t):= \int_{J}E(X_{t}X_{s})f(s)ds 
        \end{equation}
        is compact and selfadjoint in the Hilbert space $L^{2}(J)$.
  \item For every partition 
        \[
          \pi:t_{0} < t_{1} < \cdots < t_{n}, \quad t_{0}=a, t_{n}=b;
        \]
        the following matrix
        \begin{equation}
        \label{eq5.2}
          (M_{J, \pi})_{i,j}:=E(X_{t_{i}} X_{t_{j}})
        \end{equation}
        offers a discrete approximation for the operator $T_{j}$ in 
        (\ref{eq5.1}).
  \item Set
        \begin{equation}
        \label{eq5.3}
          \mathcal{H}(J):=L^{2}(J)\ominus \text{ker}(T_{J})=
          \{g \in L^{2}(J): \langle g,k\rangle_{L^{2}}=0, \forall k \in 
          \text{ker}(T_{J})\}.
        \end{equation}
        Then an application of the spectral theorem to $T_{J}$ yields the 
        following sequence of orthogonal eigenfunctions 
        $\varphi_{1}, \varphi_{2}, \cdots $ in $\mathcal{H}(J)$, and numbers 
        $\lambda_{1}, \lambda_{2}, \cdots \in \mathbb{R}_{+}$ such that 
        $\lambda_{k} \to 0$; 
        \[
          \lambda_{1} \geq \lambda_{2} \geq \cdots > 0, \quad 
          \lambda_{k} \to 0
        \]
        such that
        \begin{equation}
        \label{eq5.4}
          T_{J}\varphi_{k}=\lambda_{k}\varphi_{k} \quad k=1,2, \cdots
        \end{equation}
        orthogonality relations in the $t-$domain:
        \begin{equation}
        \label{eq5.5}
          \langle \varphi_{j}, \varphi_{k}\rangle_{L^{2}(J)} 
          = \int_{J}\overline{\varphi_{j}}{\varphi_{k}}dm = \delta_{j,k};
        \end{equation}
        and the closed span of $\{\varphi_{k}\}$ is $\mathcal{H}(J)$.
  \item Set 
        \begin{equation}
        \label{eq5.6}
          Z_{k}(\cdot):=\frac{1}{\sqrt{\lambda_{k}}} 
          \int_{J}\overline{\phi_{k}(t)}X_{t}(\cdot)dt,
        \end{equation}
        and note that each $Z_{k}$, $k=1,2, \cdots $ is a random variable,
        \[
          Z_{k} \in L^{2}(\Omega, \mathcal{S}, P).
        \]
        Moreover, a calculation yields (orthogonality relations in the 
        $\omega-$domain:
        \begin{equation}
        \label{eq5.7}
          E(\overline{Z_{j}}Z_{k})=\delta_{j,k}
        \end{equation}
  \item Aside; note that if $(X_{t})$ is assumed Gaussian, then each $Z_{k}$, 
        $k=1,2, \cdots$ is Gaussian as well.
  \item Karhunen-Lo\`{e}ve, or Generalized Fourier Expansion: \\
        In $L^{2}(J \times \Omega, m \times P)$, we have the following 
        pointwise a. e. representation
        \begin{equation}
        \label{eq5.8}
          X(t, \omega)=
          \sum_{k=1}^{\infty}\sqrt{\lambda_{k}}\varphi_{k}(t)Z_{k}(\omega), 
        \end{equation}
        as well as
        \begin{equation}
        \label{eq5.9}
          \lim_{N\to \infty} \left\Vert X(\cdot, \cdot) - 
          \sum_{k=1}^{N}\sqrt{\lambda_{k}}\varphi_{k}(\cdot)Z_{k}(\cdot)
          \right\Vert_{L^{2}(m \times P)} =0.
        \end{equation}

\end{itemize}

\subsection{Part B}
\label{partB}

We now turn to the details of the proof of (\ref{eq2.4}) and (\ref{eq3.1}).

Writing out equation (\ref{eq5.4}), we get
\begin{equation}
\label{eq5.10}
  \int_{J}E(X_{t}X_{s})\varphi_{k}(s)ds=\lambda_{k}\varphi_{k}(t);
\end{equation}
and so from the assumptions (i)-(iii) and eq. (\ref{eq2.1}) we conclude that 
each of the eigenfunctions $\varphi_{k}(\cdot)$ is continuous and of bounded 
variation.

This means that whenever $t_{0}, t \in J$, i.e., $a\leq t_{0}<t\leq b$, the 
expression
\begin{equation}
\label{eq5.11}
  \int_{t_{0}}^{t}f(s)d\varphi_{k}(s)
\end{equation}
is a well defined Stieltjes integral.  Morever, if $f$ is assumed of bounded 
variation, 
\begin{equation}
\label{eq5.12}
  \int_{t_{0}}^{t}fd\varphi_{k}=[f\varphi_{k}]_{t_{0}}^{t}-
  \int_{t_{0}}^{t}\varphi_{k}(s)f'(s)ds.
\end{equation}

We now turn to the approximation (\ref{eq2.4}) from the theorem, and we use the 
Karhunen-Lo\`{e}ve expansion (\ref{eq5.8}) in the computation of 
\begin{equation}
\label{eq5.13}
  \Delta X_{t_{i}}:=X_{t_{i+1}}-X_{t_{i}}
\end{equation}
for a fixed (chosen) partition $\pi$ as specified in (\ref{eq2.2}).

Using condition (3)(ii) in the statement of the theorem, we note that for 
fixed $J$, the operator $T_{j}$ in (\ref{eq5.1}) is trace-class.  From 
operator theory (Mercer's theorem), we know that 
\begin{equation}
\label{eq5.14}
  \text{trace}(T_{J})=\int_{J}E(X_{t}^{2})dt
  =\sum_{k=1}^{\infty}\lambda_{k} < \infty
\end{equation}
i.e., integration in $(\ref{eq5.1})$ over the diagonal $s=t$.  And so in 
particular, finiteness of $\sum_{k=1}^{\infty}\lambda_{k}$ follows.

In the study of the operator $T_{J}$ from (\ref{eq5.1}) we make use of 
tools from Hilbert space theory of integral operators. In particular, in 
the estimate (\ref{eq5.14}) we use Mercer's theorem. However in applications 
to covariance kernels (\ref{eq2.1}) one often has stronger properties. It is 
known that if the kernel in (\ref{eq2.1}) is Lipschitz of degree $\gamma$ 
with  $\gamma > \frac{1}{2}$  in one of the two variables (with the 
other fixed), then the operator  $T_{J}$ in (\ref{eq5.1}) will 
automatically be nuclear. For the literature on this we refer to  
\cite{Do93, Ku83,  LL52,  St58}.  We further note that this Lipschitz 
condition is indeed satisfied for the covariance kernel of fractional 
Brownian motion, see e.g., \cite{IA07}.

Set $(\Delta \varphi_{k})_{t_{i}}:=\varphi_{k}(t+\Delta t)-\varphi_{k}(t)$.
Using the Hilbert space $L^{2}(J \times \Omega, m \times P)$ and its 
tensor-product representation, 
$L^{2}(J)\otimes L^{2}(\Omega, \mathcal{S}, P)$, we get 
\begin{align*}
  \sum_{t=0}^{n-1}f(t_{i})\Delta X_{t_{i}}(\omega) 
  &\underset{\text{by } (\ref{eq5.8})}{=} 
  \sum_{i=0}^{n-1}\sum_{k=1}^{\infty}f(t_{i})\sqrt{\lambda_{k}}
  (\varphi_{k}(t_{i+1})-\varphi_{k}(t_{i}))Z_{k}(\omega) \\
  &= \sum_{k=1}^{\infty}(\sum_{i=0}^{n-1}f(t_{i})(\Delta \varphi_{k})_{t_{i}})
     \sqrt{\lambda_{k}}Z_{k}(\omega)\\
\end{align*}
and therefore
\begin{equation}
\label{eq5.15}
  E(|S_{\pi}(f, X)|^{2}) \underset{\text{by } (\ref{eq5.7})}{=}
  \sum_{k=1}^{\infty}\lambda_{k}\left\vert 
  \sum_{i=0}^{n-1} f(t_{i})(\Delta \varphi_{k})_{t_{i}} \right\vert^{2}.
\end{equation}
Since $f$ is assumed contions, and each $\varphi_{k}$ is of bounded variation, 
the following convergence holds:
\begin{equation}
\label{eq5.16}
  \lim_{|\pi|\to 0} \sum_{i=0}^{n-1}f(t_{i})(\Delta \varphi_{k})_{t_{i}} 
  = \int_{t_{0}}^{t}fd\varphi_{k}.
\end{equation}

Now if the function $f$ is satisfying $f \in L^{2}(J)$ and $f' \in L^{2}(J)$, 
then we get the following estimate, relying on the boundary representation 
(\ref{eq5.12}) and Parseval, see also (\ref{eq5.5}): For the RHS 
in (\ref{eq5.15}) 
we have; after passing to the limit:
\begin{align*}
  \sum_{k=1}^{\infty}\lambda_{k} \left\vert 
  \int_{t_{0}}^{t}f(s)d\varphi_{k}(s)\right\vert^{2} 
  &\underset{\text{by } (\ref{eq5.4})}{=}
  \text{(boundary terms)} + \sum_{k=1}^{\infty}\lambda_{k} \left\vert 
  \int_{t_{0}}^{t}\varphi_{k}(s)f'(s)ds\right\vert^{2}  \\
  &\leq \text{(boundary terms)} + \lambda_{1}\sum_{k=1}^{\infty} \left\vert 
  \int_{t_{0}}^{t}\varphi_{k}f'ds\right\vert^{2}  \\
  &\underset{\text{by } (\ref{eq5.5}) \text{ and Parseval}}{\leq}
  \text{(boundary terms)} + \lambda_{1} \int_{t_{0}}^{t}|f'(s)|^{2}ds;
\end{align*}
which is the desired conclustion in part (c) of the theorem.

\begin{proof}
Of Corollary \ref{C:3.2}: The essential point is formular (\ref{eq5.8}).  
However, in substitution of the expression on the RHS in (\ref{eq5.8}) we 
make use of double-orthogonality, i.e., (\ref{eq5.5}) and (\ref{eq5.7}). 
Specifically, we have the tensor product representation
$L^{2}(J \times \Omega, m \times P)=L^{2}(J) \otimes L^{2}(\Omega)$,
and so $X=\sum_{k=1}^{\infty}\sqrt{\lambda_{k}}\varphi_{k} \otimes Z_{k}$
in (\ref{eq5.8}) refers to the tensor representation.
\end{proof}

\section{Entropy: Optimal Bases}
\label{sec:6}
In this section we compare the choice of ONB from section \ref{sec:3} with 
alternative choices of ONBs. The application of Karhunen-Lo\`{e}ve dictates 
a particular choice of ONB.  

Historically, the Karhunen-Lo\`{e}ve arose as a tool from the interface of 
probability theory and information theory; see details with references 
inside the paper. It has served as a powerful tool in a variety of 
applications; starting with the problem of separating variables in 
stochastic processes, say $X_{t}$; processes that arise from statistical 
noise, for example from fractional Brownian motion. Since the initial 
inception in mathematical statistics, the operator algebraic contents of 
the arguments have crystallized as follows: starting from the process 
$X_{t}$,  for simplicity assume zero mean, i.e., $E(X_{t}) = 0$; create a 
correlation matrix  $T_{J}(s,t) = E(X_{s} X_{t})$. (Strictly speaking, it is 
not a matrix, but rather an integral kernel. Nonetheless, the matrix 
terminology has stuck.) The next key analytic step in the Karhunen-Lo\`{e}ve 
method is to then apply the Spectral Theorem from operator theory to a 
corresponding selfadjoint operator, or to some operator naturally associated 
with the integral kernel: Hence the name, the Karhunen-Lo\`{e}ve 
Decomposition (KLC). In 
favorable cases (discrete spectrum), an orthogonal family of functions 
$(\varphi_{n}(t))$ in the time variable arise, and a corresponding family of 
eigenvalues. We take them to be normalized in a suitably chosen 
square-norm. By integrating the basis functions $\varphi_{n}(t)$ against 
$X_{t}$, we get a sequence of random variables $Z_{n}$. It was the insight of 
Karhunen-Lo\`{e}ve \cite{Loe52} to give general conditions for when this 
sequene of random variables is independent, and to show that if the 
initial random process $X_{t}$ is Gaussian, then so are the random 
variables $Z_{n}$. 

Below, we take advantage of the fact that Hilbert space and operator 
theory form the common language of both quantum mechanics and of signal/image 
processing. Recall first that in quantum mechanics, (pure) states as 
mathematical entities ``are" one-dimensional subspaces in complex Hilbert 
space $\mathcal{H}$, so we may represent them by vectors of norm one. 
Observables ``are" selfadjoint operators in $\mathcal{H}$, and the 
measurement problem entails von Neumann's spectral theorem applied to the 
operators.

In signal processing, time-series, or matrices of pixel 
numbers may similarly be realized by vectors in Hilbert space 
$\mathcal{H}$. The probability distribution of quantum mechanical 
observables (state space $\mathcal{H}$) may be represented by choices of 
orthonormal bases (ONBs) in $\mathcal{H}$ in the usual way (see e.g., 
\cite{Jor06}).In the 1940s, Kari Karhunen (\cite{Kar46}, \cite{Kar52}) 
pioneered the use of spectral theoretic methods in the analysis of time 
series, and more generally in stochastic processes. It was followed up by 
papers and books by Michel Lo{\`e}ve in the 1950s \cite{Loe52}, and in 1965 
by R.B. Ash \cite{Ash90}. (Note that this theory precedes the surge in the 
interest in wavelet bases!)

Parallel problems in quantum mechanics and in signal processing 
entail the choice of ``good" orthonormal bases (ONBs). One particular such ONB 
goes under the name ``the Karhunen-Lo\`{e}ve basis." We will show that it is 
optimal.

\begin{definition}
\label{D:6.1}
Let $\mathcal{H}$ be a Hilbert space.
Let $(\psi_{i})$ and $(\varphi_{i})$ be orthonormal bases (ONB).
If $(\psi_{i})_{i \in I}$ is an ONB, we set $Q_{n}:=$ the orthogonal projection 
onto $span\{\psi_{1}, ... , \psi_{n}\}$.
\end{definition}

We now introduce a few facts about operators which will be needed in the paper. In particular we recall Dirac's terminology \cite{Dir47} for rank-one operators in Hilbert space. While there are alternative notation available, Dirac's bra-ket terminology is especially efficient for our present considerations.
\begin{definition}
\label{D:6.2}
Let vectors $u$, $v \in \mathcal{H}$.  Then  
\begin{equation}
\label{E:dirac1}
  \langle u | v \rangle = \text{inner product} \in \mathbb{C},
\end{equation}

\begin{equation}
\label{E:dirac2}
  |u \rangle \langle v| = \text{rank-one operator}, 
  \mathcal{H} \to \mathcal{H},
\end{equation}
 where the operator $|u \rangle \langle v|$ acts as follows

\begin{equation}
\label{E:dirac3}
|u \rangle \langle v|w = |u\rangle \langle v|w \rangle 
  = \langle v | w \rangle u, \quad \textit{for all } w \in \mathcal{H}.
\end{equation}
\end{definition}

Dirac's bra-ket and ket-bra notation is is popular in physics, and it is 
especially convenient in working with rank-one operators and inner products. 
For example, in the middle term in eq (\ref{E:dirac3}), the vector $u$ is 
multiplied by a scalar, the inner product; and the inner product comes about 
by just merging the two vectors.

\begin{definition}
\label{D:6.3}
If $S$ and $T$ are bounded operators in $\mathcal{H}$, in $B(\mathcal{H})$, 
then 
\begin{equation}
\label{E:fact3}
S|u \rangle \langle v|T = |Su \rangle \langle T^{*}v|
\end{equation}

If $(\psi_{i})_{i \in \mathbf{N}}$ is an ONB then the projection 
\[
  Q_{n}:= \text{proj span}\{\psi_{1}, ... , \psi_{n}\}
\]
is given by
\begin{equation}
\label{E:fact4}
  Q_{n} = \sum_{i=1}^{n}|\psi_{i} \rangle \langle \psi_{i}|; 
\end{equation}
and for each $i$, $|\psi_{i} \rangle \langle \psi_{i}|$ is the projection onto 
the one-dimensional subspace $\mathbf{C} \psi_{i} \subset \mathcal{H}$.
\end{definition}

\begin{definition}
\label{D:6.4}
Suppose $X_{t}$ is a stochastic process indexed by $t$ in a finite interval 
$J$, and taking values in $L^{2}(\Omega, P)$ for some probability space 
$(\Omega, P)$. Assume the normalization $E(X_{t})=0$. Suppose the integral 
kernel $E(X_{t} X_{s})$ can be 
diagonalized, i.e., suppose that 
\[
  \int_{J}{E(X_{t} X_{s})\varphi_{k}(s)}ds=\lambda_{k}\varphi_{k}(t)
\]
with an ONB $(\varphi_{k})$ in $L^{2}(J)$.  If $E(X_{t})=0$ then
\[
  X_{t}(\omega)=\sum_{k}\sqrt{\lambda_{k}}\phi_{k}(t)Z_{k}(\omega), 
  \quad \omega \in \Omega
\]
where $E(Z_{j} Z_{k})=\delta_{j,k}$, and $E(Z_{k})=0$.
The ONB $(\varphi_{k})$ is called the \textit{KL-basis} with respect to the 
stochastic processes $\{X_{t}: t \in J \}$.
\end{definition}

\begin{theorem}
\label{T:6.1}
(See \cite{JS07}) The Karhunen-Lo\`{e}ve ONB gives 
the smallest error terms in the approximation to a frame operator.
\end{theorem}
\begin{proof}
Given the operator $T_{J}$ which is trace class and positive 
semidefinite, we may apply the spectral theorem to it. What results is 
a discrete spectrum, with the
natural order $\lambda_{1} \geq \lambda_{2} \geq ...$ and a corresponding ONB 
$(\varphi_{k})$ consisting of eigenvectors, i.e., 
\begin{equation}
\label{E:eigenvector}
T_{J}\varphi_{k}=\lambda_{k}\varphi_{k}, k \in \mathbb{N}
\end{equation}
called the Karhunen-Lo\`{e}ve data.  The spectral data may be constructed 
recursively starting with 
\begin{equation}
\label{E:lambda1}
\lambda_{1}=\underset{\varphi \in \mathcal{H}, \|\varphi\|=1}{sup}
\langle \varphi| T_{J}\varphi \rangle 
= \langle \varphi_{1}| T_{J}\varphi_{1} \rangle
\end{equation}
and
\begin{equation}
\label{E:lambdak+1}
\lambda_{k+1}=\underset{\underset{\varphi \bot \varphi_{1}, \varphi_{2}, ..., \varphi_{k}}
{\varphi \in \mathcal{H}, \|\varphi\|=1}}{sup}
\langle \varphi| T_{J}\varphi \rangle 
= \langle \varphi_{k+1}| T_{J}\varphi_{k+1} \rangle
\end{equation}
Now an application of \cite{ArKa06}; Theorem 4.1 yields 
\begin{equation}
\label{E:ineq}
\sum_{k=1}^{n}\lambda_{k} \geq \text{tr}(Q_{n}^{\psi}T_{J}) 
  = \sum_{k=1}^{n}\langle \psi_{k}| T_{J}\psi_{k} \rangle \quad \text{for all }n,
\end{equation} 
where $Q_{n}^{\psi}$ is the sequence of projections from (\ref{E:fact4}), 
deriving from some ONB $(\psi_{i})$ and arranged such that
\begin{equation}
\label{E:psi_ineq}
\langle \psi_{1} | T_{J}\psi_{1} \rangle \geq \langle \psi_{2} | T_{J}\psi_{2} \rangle 
\geq ... \quad \text{.} 
\end{equation}
Hence we are comparing ordered sequences of eigenvalues with sequences of 
diagonal matrix entries.

Finally, we have
\[
\text{tr }(T_{J})=\sum_{k=1}^{\infty} \lambda_{k} 
=\sum_{k=1}^{\infty} \langle \psi_{k} | T_{J}\psi_{k} \rangle < \infty.
\]

The assertion in Theorem \ref{T:6.1} is the validity of
\begin{equation}
\label{E:error_ineq}
E_{n}^{\varphi} \leq E_{n}^{\psi}
\end{equation}
for all $(\psi_{i}) \in ONB(\mathcal{H})$, and all $n=1,2,...$; and moreover, 
that the infimum on the RHS in (\ref{E:error_ineq}) is attained for the 
KL-ONB $(\varphi_{k})$. But we see that (\ref{E:error_ineq}) is equivalent 
to the system (\ref{E:ineq}) in the Arveson-Kadison theorem.
\end{proof}

The Arveson-Kadison theorem is the assertion (\ref{E:ineq}) for trace class operators, see e.g., refs \cite{Arv06} and \cite{ArKa06}. That (\ref{E:error_ineq}) is equivalent to (\ref{E:ineq}) follows from the definitions. 

Our next theorem gives Karhunen-Lo\`{e}ve optimality for sequences of entropy
numbers.
\begin{theorem}
\label{T:6.2}
(See \cite{JS07}) The Karhunen-Lo\`{e}ve ONB gives the smallest sequence 
of entropy numbers in the approximation.
\end{theorem}
\begin{proof}
We begin by a few facts about entropy of trace-class operators $T_{J}$.  The entropy
is defined as
\begin{equation}
\label{E:entropy}
S(T_{J}):= -\text{tr}(T_{J}\log{T_{J}}).
\end{equation}
The formula will be used on cut-down versions of an initial operator 
$T_{J}$.  In some cases only the cut-down might be trace-class.
Since the Spectral Theorem applies to $T_{J}$, the RHS in (\ref{E:entropy}) is also
\begin{equation}
\label{E:entropy_sum}
S(T_{J})=-\sum_{k=1}^{\infty}\lambda_{k}\log{\lambda_{k}}.
\end{equation}
For simplicity we normalize such that $1=\text{tr}T_{J}=\sum_{k=1}^{\infty}\lambda_{k}$,
and we introduce the partial sums
\begin{equation}
\label{E:entropy_part_sum1}
S_{n}^{KL}(T_{J}):=-\sum_{k=1}^{n}\lambda_{k}\log{\lambda_{k}}.
\end{equation}
and
\begin{equation}
\label{E:entropy_part_sum2}
S_{n}^{\psi}(T_{J}):=-\sum_{k=1}^{n}\langle \psi_{k}|T_{J}\psi_{k}\rangle 
\log{\langle \psi_{k}|T_{J}\psi_{k}\rangle}.
\end{equation}

Let $(\psi_{i}) \in ONB(\mathcal{H})$, and set 
$d_{k}^{\psi}:=\langle \psi_{k}|T_{J}\psi_{k}\rangle$; then the inequalities 
(\ref{E:ineq}) take the form:
\begin{equation}
\label{E:d_lambda}
  \text{tr}(Q_{n}^{\psi}T_{J})=\sum_{i=1}^{n}d_{i}^{\psi} \leq \sum_{i=1}^{n} \lambda_{i}, \quad n=1,2,...
\end{equation}
where as usual an ordering
\begin{equation}
\label{E:d_order}
d_{1}^{\psi} \geq d_{2}^{\psi} \geq ...
\end{equation}
has been chosen.

Now the function $\beta(t):=t\log{t}$ is convex.  And application of Remark 6.3 in
\cite{ArKa06} then yields 
\begin{equation}
\label{E:beta_ineq}
\sum_{i=1}^{n}\beta(d_{i}^{\psi}) \leq \sum_{i=1}^{n}\beta(\lambda_{i}), 
\quad n=1,2,... \quad \text{.}
\end{equation}
Since the RHS in (\ref{E:beta_ineq}) is $-\text{tr}(T_{J}\log{T_{J}})=-S_{n}^{KL}(T_{J})$, the desired
inequalities
\begin{equation}
\label{E:S_ineq}
S_{n}^{KL}(T_{J}) \leq S_{n}^{\psi}(T_{J}), \quad n=1,2, ...
\end{equation}
follow. i.e., the KL-data minimizes the sequence of entropy numbers.
\end{proof}

\subsection*{Acknowledgment}
\label{sec:ack}
The present work was motivated by details form a graduate course taught by 
the first named author on stochastic integration and its applications. We 
are grateful to the students in the course, especially Mr Yu Xu,  for their 
comments and their inspiration.



\newcommand{\etalchar}[1]{$^{#1}$}


\end{document}